\documentclass[letterpaper, 10 pt, conference]{ieeeconf}
\usepackage{cite}
\usepackage{amsthm}
\usepackage{amssymb}
\usepackage{amsfonts}
\usepackage{mathtools}
\usepackage{algorithm}
\usepackage{algpseudocode}
\usepackage{graphicx}
\usepackage{textcomp}
\usepackage{threeparttable}
\usepackage{hyperref}
\usepackage{amssymb}
\let\labelindent\relax
\usepackage{enumitem}
\usepackage{tabu}
\usepackage[utf8]{inputenc}
\usepackage{booktabs} % for \toprule, \midrule, \bottomrule
\usepackage{setspace}
\usepackage[font=small,labelfont=bf]{caption}
\usepackage[font=small,labelfont=bf]{subcaption}
\usepackage{multirow}
\usepackage{tabularx}
\usepackage{array}
\usepackage{makecell}
\usepackage{dblfloatfix}
\usepackage{tikz-network}
\newtheorem{theorem}{\textbf{Theorem}}
\newtheorem{lemma}{\textbf{Lemma}}

\newtheorem{remark}{\textbf{Remark}}
\newtheorem{assumption}{\textbf{Assumption}}
\usepackage{authblk}
\usepackage{float}
\usepackage[table]{xcolor}
\newtheorem{definition}{\textbf{Definition}}

\usepackage{color}
\usepackage{tikz}  
\usetikzlibrary{arrows.meta, positioning, shapes, calc}
\IEEEoverridecommandlockouts

\usepackage{siunitx}
\usepackage{xcolor}
\usepackage{pifont}
\newcommand{\R}{\mathbb{R}}
\newcommand{\diag}{\operatorname{diag}}

\title{ Delay-Independent Safe Control with Neural Networks: Positive Lur’e Certificates for Risk-Aware Autonomy}
\author{Hamidreza Montazeri$^{1}$, Milad Siami$^{1}$\thanks{$^1$ Department of Electrical \& Computer Engineering, Northeastern University, Boston, MA 02115, USA.
(e-mails: {\tt\footnotesize	\{montazerihedesh.h, m.siami\}@northeastern.edu}).}}

\begin{document}
\maketitle
\begin{abstract}
We present a risk-aware safety certification method for autonomous, learning-enabled control systems. Focusing on two realistic risks, state/input delays and interval-matrix uncertainty, we model the neural network (NN) controller with local sector bounds and exploit positivity structure to derive linear, delay-independent certificates that guarantee local exponential stability across admissible uncertainties. To benchmark performance, we adopt and implement a state-of-the-art IQC NN-verification pipeline. On representative cases, our positivity-based tests run orders of magnitude faster than SDP-based IQC while certifying regimes the latter cannot—providing scalable safety guarantees that complement risk-aware control.
% We present a risk-aware safety certification method for autonomous, learning-enabled control system. Two different realistic risk sources are studied--state/input Delay and interval-matrix uncertainty. By combining local sector bounds on the NN with positivity characteristics, we obtain linear, delay-independent certificates that guarantee local exponential stability for a complex nonlinear system across admissible uncertainties. We further adopt the state-of-the-art NN verification method, to build a benchmark for comparison against our method. Comparison shows our positivity-based tests run orders-of-magnitude faster than SDP-based IQC while certifying cases the latter cannot.
% The results deliver lightweight closed-loop safety and robustness guarantees that complement perception and planning modules in autonomy stacks.
\end{abstract}

\section{Introduction}
Autonomous systems increasingly rely on learning-enabled components \cite{du2023can}, yet verifying closed-loop safety under realistic risks (e.g., time delays and parametric uncertainty) remains a bottleneck in deploying AI-enabled control in safety-critical applications, including autonomous vehicles, aerospace systems, and medical devices \cite{sznaier2022role}.
% The integration of neural networks (NNs) into safety-critical control systems has emerged as a transformative paradigm in autonomous systems, offering unprecedented capabilities in handling complex, nonlinear dynamics and adapting to uncertain environments \cite{brunke2022safe,cheng2019end}.
% However, this integration introduces significant challenges for safety certification, particularly when these learning-enabled controllers must operate under realistic operational risks such as time delays and parametric uncertainties \cite{dean2020sample,fulton2018safe}. The verification of such systems remains a critical bottleneck in deploying AI-enabled control in safety-critical applications including autonomous vehicles, aerospace systems, and medical devices \cite{dawson2023safe}.

% \subsection{Neural Network Verification in Control Systems}

Recent advances in NN verification have produced various approaches for certifying closed-loop stability. Semidefinite programming (SDP) based methods \cite{fazlyab2020safety}, barrier function learning methods \cite{qin2021learning}, among others, are presented in a survey \cite{dawson2023safe}.  
% particularly those employing integral quadratic constraints (IQCs), have shown promise in providing formal guarantees \cite{fazlyab2022safety,fazlyab2019efficient}.
On a notable instance, Fazlyab et al. \cite{fazlyab2020safety} developed a comprehensive framework using quadratic constraints for safety verification, while Yin et al. \cite{yin2021stability} extended these techniques specifically for systems with NN controllers. However, these SDP-based approaches suffer from computational complexity that scales poorly with system dimension and network size \cite{fazlyab2020safety}.
Alternative verification methods have emerged from the machine learning community, including reachability analysis \cite{xiang2018output}, abstract interpretation \cite{singh2019abstract}, and convex relaxation techniques such as CROWN \cite{zhang2018efficient}. While these methods excel at providing input-output bounds for NNs, they typically focus on open-loop properties and struggle to incorporate closed-loop dynamics, particularly under realistic operational conditions like delays and uncertainties \cite{everett2021efficient}.

% \subsection{Time-Delay Systems and Uncertainty}

Time delays are ubiquitous in networked control systems, arising from communication latencies and computational delays \cite{richard2003time}. The stability analysis of delayed systems has a rich history in control theory \cite{gu2004stability}. However, when combined with NN controllers, the analysis becomes significantly more complex. Recent works such as \cite{sun2021adaptive,10834551} addressed NN control of time-delay systems but relied on computationally expensive LMIs and complex control schemes.
Parametric uncertainty presents another critical challenge in real-world deployments. The combination of delays and uncertainties has been studied extensively for linear systems \cite{7353149,shukla2025predefined}, but the introduction of NN controllers creates a three-way interaction that existing methods struggle to address.

% \subsection{Positive Systems Theory}
Positive systems theory offers unique analytical advantages through its structural properties \cite{farina2011positive}. Metzler matrices provide computationally efficient certificates that are often delay-independent \cite{shafai2014positive}.
% Recent work has explored positive systems in various contexts, including epidemic networks \cite{darabi2022stability}, multi-agent systems \cite{ren2008distributed}, and biological networks \cite{angeli2003monotone}.
The connection between positive systems and sector-bounded nonlinearities through Lur'e systems has been established \cite{bill2016stability}, but its application to NN verification remains underexplored \cite{Hamidrezaglobalsectorbound}.

% \subsection{Contributions and Paper Organization}

This paper bridges the gap between positive systems theory and NN verification by developing a novel framework that exploits positivity constraints for efficient safety certification. Our key contributions are:

% \begin{enumerate}
\textbf{1) Local Sector Bounds for FFNNs}: We develop a novel sector bound for feedforward NNs that, unlike existing elementwise bounds (CROWN, IBP), provides network-level sector descriptions suitable for Lur'e system analysis.

\textbf{2) Positivity-Based Verification Framework}: We introduce a scalable, delay-independent certificate that leverages Metzler matrix and positive system properties to verify NN feedback loops under both time delays and interval uncertainties.

\textbf{3) Comprehensive Risk Analysis}: We provide unified treatment of three risk configurations: (i) interval uncertainty only, (ii) time delays only, and (iii) combined delays and uncertainties, demonstrating that our approach handles all cases with a single, computationally efficient framework.

\textbf{4) Comparative Evaluation}: We develop and adapt the state-of-the-art IQC-based verification method for NN feedback loops to our specific risk setting and use it as a baseline. We benchmark our method against this IQC baseline, demonstrating orders-of-magnitude speedups and certification in regimes where SDP-based approaches fail.
% \end{enumerate}

The significance of this work extends beyond theoretical contributions. As autonomous systems increasingly operate in uncertain, networked environments, the ability to efficiently certify safety under realistic operational conditions becomes paramount. Our positivity-based approach offers a path toward real-time, online verification that could enable adaptive safety certificates and risk-aware control strategies.

% The remainder of this paper is organized as follows: Section II formally states the problem and motivates our approach. Section III provides necessary background on positive systems and NN sector bounds. Section IV presents our main results, including the local sector bound construction and positivity-based verification theorems for all three risk configurations. Section V demonstrates our approach through numerical examples and comparative studies. Section VI concludes with discussions on future directions.
% \pagebreak

\subsection{Notation}
We denote by $\R^{n}_{+}$ and $\R^{n\times m}_{+}$ vectors and matrices with nonnegative real entries. The orders $ \ge ,\:>$ denote the elementwise inequality for vectors and matrices of the same size, while $\prec,\preceq$ is used in the definiteness sense. A square matrix is Metzler if its off-diagonal entries are nonnegative.
% $\mathcal{C}([-\tau,0];\mathbb{R}^n)$ is the Banach space of all continuous maps $\varphi:[-\tau,0]\to\mathbb{R}^n$ with the sup norm.
% \pagebreak

\section{Problem Statement and Motivation}
\subsection{Problem Description}\label{sec:problem}
We study the stability of an autonomus AI-enabled system under two sources of risk, namely, delay and interval uncertainty. We model the system in the form of a classical Lur'e system where a static NN is in the feedback loop of a linear time invariant (LTI) plant subjected to interval uncertainty and state and input delay written in the form:
% \begin{subequations}
\begin{align}\label{eq:wholesys}\small
    \dot x(t) &= [\underline{A}_0,\overline{A}_0]\,x(t)
    + \sum_{i=1}^{\ell} [\underline{A}_i,\overline{A}_i]\,x(t-\tau_i)
    + B_i u(t-\tau_i),\nonumber\\
    y(t) &= Cx(t), \qquad
    u(t) = \Phi\big(y(t)\big).
\end{align}
% \end{subequations}
Here, \(A_0, A_i \in \mathbb{R}^{n\times n}\) are constant and lie in the elementwise intervals \([\underline{A}_0,\overline{A}_0]\) and \([\underline{A}_i,\overline{A}_i]\); \(B_i \in \mathbb{R}^{n\times m}\) and \(C \in \mathbb{R}^{p\times n}\) are the input and output matrices; and \(x(t) \in \mathbb{R}^n\), \(u(t) \in \mathbb{R}^m\), \(y(t) \in \mathbb{R}^p\) denote the state, input, and output, with initial state \(x(0)\coloneqq x_0\). Assume an initial history \(x_{h_i}(t)\) on \(t\in [-\tau_{i},0]\) so all delayed terms are well-defined.
The controller is a static nonlinear mapping $\Phi:\mathbb{R}^{p}\to\mathbb{R}^m$ 
with $\Phi(0)=0$ rendering the origin $(x,u)=(0,0)$ to be equilibrium.
The function $\Phi$ is realized by a $q$-layer FFNN with the following equations:
\vspace{-.2cm}

{\small \begin{align}\label{eq:ffnn}
&\nu^{(1)} = W^{(1)} y + b^{(1)}, 
\qquad \omega^{(1)} = \varphi^{(1)}\!\big(\nu^{(1)}\big),\nonumber\\
&\nu^{(i)} = W^{(i)} \omega^{(i-1)} + b^{(i)}, 
\quad \omega^{(i)} = \varphi^{(i)}\!\big(\nu^{(i)}\big), \quad i=2,\dots,q,\nonumber\\
&u = W^{(q+1)} \omega^{(q)} + b^{(q+1)}.
\end{align}}

Dimensions are $W^{(i)}\!\in\mathbb{R}^{n_i\times n_{i-1}}$ and $b^{(i)}\in\R^{n_i}$, with the dimensions of the first and last layers being fixed to the dimensions of system's input and output. Each activation $\varphi^{(i)}(\cdot)$ acts elementwise and is assumed to be monotone on a prescribed interval. Most widely used activation functions satisfy the requirement.
We analyze the equilibrium at the origin by enforcing $\Phi(0)=0$ either by design, $b^{(i)}=0$, or by shifting coordinates around the operating point.
The set \((x^*,u^*,\nu^*,\omega^*)\) denotes the equilibrium states.

\subsection{Challenges and Why Existing Methods Fall Short}
The closed loop system in \eqref{eq:wholesys} is a complex nonlinear system with three coupled sources of uncertainty/nonlinearity, namely time delay, interval matrices, and an NN controller, making safety verification challenging. Standard nonlinear analysis reduces the problem to large SDPs to certify delay/uncertainty margins. Inserting an NN further inflates these LMIs with auxiliary variables and multipliers. Feasibility is also highly sensitive to the specific delay and uncertainty values, so the SDP must be re-solved whenever these parameters change. As a result, SDP-based approaches become computationally intractable at scale \cite{fazlyab2020safety}. Even LMI-free analytic results for linear systems 
% \cite{darabi2022stability}
are delay-dependent and do not extend gracefully to additional uncertainties or nonlinear (NN) feedback. We later illustrate these limitations on representative examples and explanatory notes.

In response, we develop a positivity-based verification framework for \eqref{eq:wholesys}. We wrap the NN in a tailored sector bound that enforces positivity of the feedback interconnection, and then leverage tools from positive-systems theory. The outcome is a delay-independent, highly scalable method for safety verification of NN feedback loops.

\section{Preliminaries and Background}
\subsection{Positive Systems}
Consider an LTI system of the form:
\begin{equation}\label{eq:generallti}\small
    \dot x(t) = A x(t) + B u(t), \quad
    y(t) = C x(t),
\end{equation}

\begin{definition}\label{def:positive system}
    The system in \eqref{eq:generallti} is called ``internally positive'' if, \( \forall t > 0 \), we have \( x(t) \geq 0 \), given \( x_0 \geq 0 \) and \( u(t) \geq 0 \).
\end{definition}
\begin{lemma}\label{lem:lemav}
    The system in \eqref{eq:generallti} is internally positive if and only if \( A \) is a Metzler matrix, $B \in \mathbb R_+^{n\times m}$, and $C \in \mathbb R_+^{p\times n}$.
    If the internally positive system is also asymptotically stable, then there exists a vector \( v \in \mathbb{R}_+^n \), with \( v > 0 \), such that \( v^\top A < 0 \).
\end{lemma}
% \subsection{Delayed Linear Metzlerian system}
Consider the delayed linear system
\begin{equation}\label{eq:generaldelaylti}\small
    \dot x(t) = A_0 x(t) + \sum_{i=1}^l A_ix(t-\tau_i)+ B_i u(t-\tau_i), \quad
    y(t) = C x(t),
\end{equation}
\begin{lemma}\label{lem:delayedpositivesys}
The system \eqref{eq:generaldelaylti} is internally positive if and only if $A_0$ is a Metzler matrix and $A_i \in \mathbb{R}_{+}^{n \times n} ; B_i \in \mathbb{R}_{+}^{n \times m}, C \in \mathbb{R}_{+}^{p \times n}$ are nonnegative matrices.
\end{lemma} 

The preceding definitions and lemmas are presented in foundational references on positive systems, e.g., \cite{farina2011positive}, and in subsequent developments such as \cite[Lem. 1]{shafai2014positive}.

\subsection{NN Sector Bounds}\label{subsec:nnpreliminaries}
% We begin by introducing a localized sector condition for nonlinearities, which plays a central role in our analysis.
\begin{definition}[$\Gamma-$Sector Bounded Function]\label{def:gammasectorbound}
    Given an interval $[\Sigma_1, \Sigma_2]$ with $\Sigma_1, \Sigma_2 \in \mathbb{R}^{m\times p}$ and $\Sigma_1 \le \Sigma_2$, a static nonlinear function $\Phi(y): \mathbb{R}^p \to \mathbb{R}^m $ is said to be $\Gamma-$sector bounded within that interval if, for a compact and connected set \( \Gamma \subseteq \mathbb{R}^p \), the following condition holds:
\begin{equation}\small
    % (\Phi(y, \cdot) - \Sigma_1 y)^\top (\Sigma_2 y - \Phi(y, \cdot)) \geq 0.
    \Sigma_1 y \leq \Phi(y) \leq \Sigma_2 y, \quad \forall y \in \Gamma \subset \mathbb{R}^p.
\end{equation}
\end{definition}
% \begin{figure*}[t]
%     \centering
%     \begin{subfigure}[t]{0.25\textwidth}
%         \centering
%         \includegraphics[width=\textwidth]{Figs/posnu.eps}
%         \caption*{(i): $\beta = \frac{\tanh(\overline{\nu})}{\overline{\nu}},\alpha = \frac{\tanh(\underline{\nu})}{\underline{\nu}},$ \\ \centering$\beta \underline{\nu} \le\tanh(\nu)\le \alpha \overline{\nu}$}
%         \label{fig:posnu}
%     \end{subfigure}
%     % \hfill
%     \begin{subfigure}[t]{0.25\textwidth}
%         \centering
%         \includegraphics[width=\textwidth]{Figs/negnu.eps}
%         \caption*{(ii): $\beta = \frac{\tanh(\underline{\nu})}{\underline{\nu}},\alpha = \frac{\tanh(\overline{\nu})}{\overline{\nu}},$ \\ \centering$\alpha \underline{\nu} \le\tanh(\nu)\le \beta \overline{\nu}$}
%         \label{fig:negnu}
%     \end{subfigure}
%     % \hfill
%     \begin{subfigure}[t]{0.25\textwidth}
%         \centering
%         \includegraphics[width=\textwidth]{Figs/unstablenode1.eps}
%         \caption*{(iii): $\alpha_1 = \alpha_2 = \frac{d\tanh(\nu)}{d(\nu)}|_{\nu=0} = 1,$\\
%         \centering $-\alpha_1|\underline{\nu}| \le \tanh{(\nu)} \le \alpha_2 |\overline{\nu}|$}
%         \label{fig:unstablenode}
%     \end{subfigure}
%     \vspace{-.2cm}
%     \caption{\small Illustration of linear relaxation of $\phi = \tanh(\nu)$ under different intervals for $\nu$.\vspace{-.6cm}}
%     \label{fig:top_three_images}
% \end{figure*}
% in preamble
% \usepackage{graphicx}

% in body
\paragraph*{\textbf{Interval Arithematic}}

To apply the local sector bounds, we first need to compute box bounds on each pre-activation $\nu^{(i)}$. These boxes are calculated by propagating the input once in the forward path of the network. Initialize at the first layer by selecting an interval $[\,\underline \nu^{(1)},\overline \nu^{(1)}\,]$ containing the equilibrium $\nu_*^{(1)}$. Monotonic activations map intervals elementwise, e.g., if $\nu\in[\underline\nu,\bar\nu]$, then $\phi(\nu)\in[\phi(\underline\nu),\phi(\bar\nu)]$, yielding $[\,\underline \omega^{(1)},\overline \omega^{(1)}\,]$. For the next layer, the $i$th coordinate ($\nu_i^{(2)} = W^{(2)}_{i,:}\omega^{(1)}+b^{(2)}_{i}$) satisfies
\[
\overline \nu^{(2)}_{i}=\max_{\underline \omega^{(1)}\le \omega^{(1)}\le \overline \omega^{(1)}} \nu_i^{(2)},\qquad
\underline \nu^{(2)}_{i}=\min_{\underline \omega^{(1)}\le \omega^{(1)}\le \overline \omega^{(1)}} \nu_i^{(2)},
\]
where $W_{i,:}$ is the $i$th row of $W$. Let $c:=\tfrac12(\overline \omega^{(1)}+\underline \omega^{(1)})$ and $r:=\tfrac12(\overline \omega^{(1)}-\underline \omega^{(1)})$. Since the objective is linear over the box $[\,\underline \omega^{(1)},\overline \omega^{(1)}\,]$, these programs have the closed forms
\vspace{-.3cm}

{\small\[
\overline v^{(2)}_{i}=W^{(2)}_{i,:} c+b^{(2)}_{i}+|W^{(2)}_{i,:}|\,r,\quad
\underline v^{(2)}_{i}=W^{(2)}_{i,:} c+b^{(2)}_{i}-|W^{(2)}_{i,:}|\,r.
\]}

Thus $\nu^{(2)}\in[\underline \nu^{(2)},\overline \nu^{(2)}]$. Repeating this propagation layer-by-layer provides $[\underline \nu^{(i)},\overline \nu^{(i)}]$ for all layers, which we then use to form the local sector bounds.

\section{Main Results}
This section is organized into three interconnected subsections. First, we develop a \emph{local sector bound} for FFNNs, providing the analytical machinery and notation used throughout the sequel. Building on this foundation, the second subsection introduces a \emph{positivity-based verification} framework, leveraging the sector bounds and positive systems theory to obtain delay and uncertainty-robust certificates. The third subsection then presents an \emph{IQC-based verification} method formulated to utilize the state-of-the-art NN-in-the-loop analysis pipelines presented in \cite{yin2021stability} for our specific system; beyond furnishing an alternative certificate, this construction is used to assemble a strong benchmark against which we systematically compare our proposed method.

% \subsection{Problem Formulation}\label{sec:problem}
% We study the stability of an NN in the feedback loop of a {\color{blue}positive} LTI plant subjected to interval uncertainty and state and input delay written in the form:
% \begin{subequations}\label{eq:wholesys}
% \begin{align}\small
%     \dot x(t) &= [\underline{A}_0,\overline{A}_0]\,x(t)
%     + \sum_{i=1}^{\ell} [\underline{A}_i,\overline{A}_i]\,x(t-\tau_i)
%     + B_i u(t-\tau_i), \label{eq:linpart}\\
%     y(t) &= Cx(t), \qquad
%     u(t) = \Phi\big(y(t)\big). \label{eq:nonlinpart}
% \end{align}
% \end{subequations}
% The controller is a static nonlinear function $\mathbb{R}^{p}\to\mathbb{R}^m$ mapping 
% % \begin{equation}\label{eq:nn}
% %     u(t) = \Phi\big(y(t)\big),
% % \end{equation}
% with $\Phi(0)=0$ rendering the origin $(x,u)=(0,0)$ to be equilibrium.
% % We denote by $\mathcal{X}$ a neighborhood of the origin where the NN is evaluated.
\begin{figure}[!t]
  \centering
  \includegraphics[width=\linewidth]{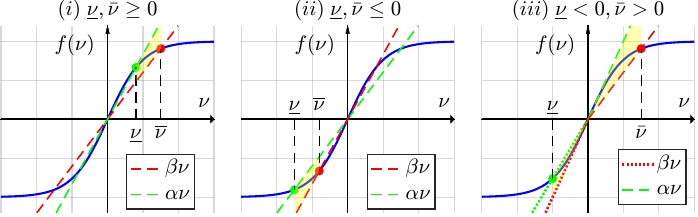} % .pdf/.png
  \caption{Sector relaxations of $\phi(\nu)=\tanh(\nu)$:\\
(i) $\beta\nu \le \tanh(\nu) \le \alpha\nu: \quad\beta=\tanh(\overline{\nu})/\overline{\nu}, \quad \alpha=\tanh(\underline{\nu})/\underline{\nu}$.\\
(ii) $\beta\nu \le \tanh(\nu) \le \alpha\nu$
$\quad \beta=\tanh(\overline{\nu})/\overline{\nu}$, $\quad \alpha=\tanh(\underline{\nu})/\underline{\nu}$
.\\
(iii) $-\beta|\nu| \le \tanh(\nu) \le \alpha|\nu|$  with
$\beta=\alpha=\tanh'(0)=1$.\vspace{-.3cm}}
  \label{fig:tanh-relaxations}
\end{figure}

\subsection{Local Sector Bounds for FFNNs}\label{sec:local-sector}
Our analysis requires a network-level sector bound for the entire NN, whereas existing bounds in the literature (e.g., CROWN, IBP, BERN-NN) generally yield elementwise or layerwise affine enclosures that do not match the required sector geometry. Accordingly, in this section we develop a novel sector bound for FFNNs. The construction follows the standard convex-relaxation pipeline popular in ML verification bounds, but replaces the usual affine envelopes with a tailored sector relaxation of the activation functions, which is then propagated through the network to produce a \([\gamma_1,\gamma_2]\) sector bound.
 % Using this definition, the following theorem establishes exponential stability of output trajectories that remain inside the defined $\Gamma$ set.
% \begin{theorem}\label{the:mathmanip}
%     Consider the Lur’e system given in \eqref{eq:positiveluresystem}, where $B, C \geq 0$ and the nonlinearity $\Phi$ is $\Gamma-$sector bounded within the interval $[\Sigma_1,\Sigma_2]$ with $y_* \in \Gamma$. Suppose that $A+B\Sigma_1C$ is Metzler and $A+B\Sigma_2C$ is Hurwitz; then all the output trajectories that remain within the set $\Gamma$ are exponentially stable.
% \end{theorem}
% For the sake of brevity, we omit the full proof here; it follows a similar reasoning to the global positive Aizerman conjecture as presented in \cite{drummond2022aizerman}, with the additional assumption that the system trajectory remains confined within the sector bounded region \( \Gamma \). Following Theorem~\ref{the:mathmanip}, Lemma~\ref{lem:mathmanip} provides a sufficient condition for identifying a subset of initial states that guarantees the output trajectory remains within $\Gamma$ set.
Consider the FFNN in \eqref{eq:ffnn} with fixed weight/bias and input
$y$ restricted to the componentwise box
$\Gamma := \{y\in\R_+^{n_0} : \underline y \le y \le \overline y,\underline y \neq 0\footnote{Since the application of the bound is in positive systems regime, usually $y\in \R^{n_0}_+$. Moreover, setting biases to zero will waive the requirement $\underline y \neq 0$.} \}$.
\subsection*{Step A: propagation through weights and biases}
Assume the post-activation vector of layer $i$ satisfies
\begin{equation}\label{eq:startWs}
    \ell^{(i)} y \;\le\; \omega^{(i)} \;\le\; u^{(i)} y,
\end{equation}
with initialization $\ell^{(0)} = u^{(0)} = I_{n_0}$. For the next
pre-activation $\nu = W\omega + b$, split $W = W^{+}+W^{-}$ with
$W^{+}_{ij}=\max\{W_{ij},0\}$ and $W^{-}_{ij}=\min\{W_{ij},0\}$, and set
\begin{align}
\underline{\sigma} &= W^{+}\ell + W^{-}u, &
\overline{\sigma}  &= W^{+}u + W^{-}\ell. \label{eq:sigma}
\end{align}
Then
\begin{equation}\label{eq:weightshare}
\underline{\sigma}\,y \;\le\; W\,\omega \;\le\; \overline{\sigma}\,y.
\end{equation}
For the bias, let $a_1:=\min_k \underline y_k$ and $a_2:=\max_k \overline y_k$,
and define $\underline\delta,\overline\delta\in\R^{n_i\times n_0}$ rowwise by
\begin{equation}\label{eq:biasshare}\small
\underline\delta_j =
\begin{cases}
 b_j/a_2 & b_j\ge 0\\
 b_j/a_1 & b_j<0
\end{cases},\qquad
\overline\delta_j =
\begin{cases}
 b_j/a_1 & b_j\ge 0\\
 b_j/a_2 & b_j<0
\end{cases}.
\end{equation}
Thus, $\underline\delta_j y \le b\le \overline\delta_j y$. Combining \eqref{eq:weightshare}–\eqref{eq:biasshare} yields
\begin{equation}\label{eq:preactivationsectrobound}\small
    L^{(i+1)}y \;\le\; \nu^{(i+1)} = W^{(i+1)}\omega^{(i)}+b^{(i+1)} \;\le\; U^{(i+1)}y,
\end{equation}
with
$L:=\underline{\sigma}+\underline{\delta}$ and $U:=\overline{\sigma}+\overline{\delta}$.
This preserves the linear-in-$y$ sector form layer by layer.

\subsection*{Step B: Activation relaxation}\label{subsec:step}
We convert $\omega^{(i)}=\phi(\nu^{(i)})$ into linear sector bounds depending on the box of preactivation $[\underline \nu^{i},\overline \nu^{i}]$ calculated by the interval arithmetic explained in subsection \ref{subsec:nnpreliminaries}.
Assume the pre-activation of layer $i$ obeys
\begin{equation}\label{eq:firstoftanh}
      L^{(i)}y \;\le\; \nu^{(i)} \;\le\; U^{(i)}y,
\end{equation}
as shown in \eqref{eq:preactivationsectrobound}. For neuron $j$, let $\underline\nu_j^{(i)}$ and $\overline\nu_j^{(i)}$ be the
corresponding entries of the preactivation box $[\underline \nu^{i},\overline \nu^{i}]$.
On the interval $[\underline\nu_j^{(i)},\overline\nu_j^{(i)}]$, define slopes
$(\alpha_j^{(i)},\beta_j^{(i)})$ so that
$\beta_j^{(i)}\nu_j^{(i)} \le \phi(\nu_j^{(i)}) \le \alpha_j^{(i)}\nu_j^{(i)}$.
% holds on that interval with $(m,M)=(\beta_j^{(i)},\alpha_j^{(i)})$:
{\small\begin{align*}
&\mathbf{(i)}\;~ \text{if }\underline\nu,\overline{\nu}\ge 0:\quad
&\beta_j^{(i)}=\frac{\phi(\overline\nu)}{\overline\nu},\;
\alpha_j^{(i)}=\frac{\phi(\underline\nu)}{\underline\nu};\\
&\mathbf{(ii)}\;~ \text{if }\underline\nu, \overline\nu\le 0:\quad
&\beta_j^{(i)}=\frac{\phi(\overline\nu)}{\overline\nu},\;
\alpha_j^{(i)}=\frac{\phi(\underline\nu)}{\underline\nu};\\
&\mathbf{(iii)}\:~ \text{if }\underline\nu\le0, \overline\nu\ge 0:\quad
&\beta_j^{(i)}= \alpha_j^{(i)}=\sup_{\nu\neq 0}|\phi(\nu)|/|\nu|.
\end{align*}}
% If the interval crosses zero (\,$\underline\nu<0<\overline\nu$\,), we use the
% odd, slope-bounded property of $\tanh$ and set
% $\alpha_j^{(i)}=1$, $\beta_j^{(i)}=-1$ and take absolute values in the shape
% rows as below. (For a general monotone $\phi$, replace $(\pm1)$ by
% $\pm L_{\max}$ where $L_{\max}:=\sup_{\nu\neq 0}|\phi(\nu)|/|\nu|$ on the interval.)

This calculation is visualized for $\phi=\tanh$ in Fig. \ref{fig:tanh-relaxations}.
Gather the neuronwise slopes into diagonal matrices
\begin{equation}\label{eq:DupandDdown}\small
\underline D^{(i)}:=\diag(\beta_j^{(i)}), \quad \overline D^{(i)}:=\diag(\alpha_j^{(i)}),
\end{equation}
and define “shape” matrices by row replacement to handle the sign-crossing case $(iii)$:
% \begin{align}\small
% \label{eq:lhatUhat}
%      & \hat L^{(i)}_{j,:} :=
%       \begin{cases}
%           L^{(i)}_{j,:}   & \text{cases } (i),(ii)\\
%           -|L^{(i)}_{j,:}| & \text{case } (iii),
%       \end{cases}\nonumber\\
%       &
%       \hat U^{(i)}_{j,:} :=
%       \begin{cases}
%           U^{(i)}_{j,:}   & \text{cases } (i),(ii)\\
%           |U^{(i)}_{j,:}| & \text{case } (iii).
%       \end{cases}
% \end{align}
\begin{equation}\small
\label{eq:lhatUhat}
     \hat L^{(i)}_{j,:} :=
      \begin{cases}
          L^{(i)}_{j,:}   & \text{case } (i),(ii)\\
          -|L^{(i)}_{j,:}| & \text{case } (iii),
      \end{cases} \quad
      \hat U^{(i)}_{j,:} :=
      \begin{cases}
          U^{(i)}_{j,:}   & \text{case } (i),(ii)\\
          |U^{(i)}_{j,:}| & \text{case } (iii).
      \end{cases}
\end{equation}
Then the post-activation is bounded by
\begin{equation}\label{eq:DandL}\small
      \underline{D}^{(i)}\,\hat L^{(i)}\,y
      \;\le\;
      \phi\!\bigl(\nu^{(i)}\bigr)
      \;\le\;
      \overline{D}^{(i)}\,\hat U^{(i)}\,y,\qquad \forall\,y\in\Gamma.
\end{equation}

\subsection*{Network-level sector bound}
Recursing Steps A–B from input to output yields a single local sector
for the whole network.

\begin{theorem}[Local Sector Bound for FFNN]\label{thm:nn_sector_bound}
Consider an FFNN as in \eqref{eq:ffnn}.
For any $y\in\Gamma=[\underline y,\overline y]\in \R^p_+ $, the NN mapping $\Phi(y)$ satisfies
\begin{equation}\label{eq:localsectorbound}
\gamma_1\, y \;\le\; \Phi(y) \;\le\; \gamma_2\, y,
\end{equation}
{\small\[
\gamma_1 = L^{(q+1)} \!\left( \prod_{i=1}^{q} \underline{D}^{(i)} \hat{L}^{(i)} \right),\quad
\gamma_2 = U^{(q+1)} \!\left( \prod_{i=1}^{q} \overline{D}^{(i)} \hat{U}^{(i)} \right),
\]}

where $L^{(i)},U^{(i)}$ come from \eqref{eq:preactivationsectrobound},
and $\underline{D}^{(i)},\overline{D}^{(i)},\hat L^{(i)},\hat U^{(i)}$ from
\eqref{eq:DupandDdown}–\eqref{eq:lhatUhat}.
\end{theorem}

% \noindent
% The pair $(\gamma_1,\gamma_2)$ is input-dependent and can be directly used,
% e.g., in our analysis in the below section.

\subsection{Positivity-Based Verification Method}
In this section, we utilize local sector bounds and positivity constraint on systems to develop a scalable delay-independent method for verification of NN feedback loops.
\begin{assumption}\label{assump1}
$\underline{A}_0$ is Metzler and $B_i,C \ge 0$.
\end{assumption}
\begin{lemma}\label{lem:lurepositivity}
    Suppose Assumption \ref{assump1} holds for system \eqref{eq:wholesys}, and the feedback nonlinearity $\Phi$ is $\Gamma-$sector bounded in $[\Sigma_1,\Sigma_2]$. Then  the Lur'e system \eqref{eq:wholesys} is internally positive in $\Gamma$, if and only if $\underline A_i + B_i\Sigma_1C \ge 0$.
\end{lemma}

\begin{proof}
To avoid notational complexity, we present the proof for a single delay channel ($\ell=1$), however, the argument extends unchanged to multiple channels.

\emph{Necessity.} Suppose $\Phi(Cx(t-\tau)) = \Sigma_1 C x(t-\tau)$ which belongs to $\text{Sector}[\Sigma_1, \Sigma_2]$. In this case, the closed-loop system reduces to an LTI system with dynamics $A_0 x(t) + \left[ A_1 +  B\Sigma_1 C\right]x(t-\tau)$. 
If the Lur’e system is positive for all admissible nonlinearities and all admissible $A_0$ and $A_i$, it must, in particular, be positive for the linear case $\underline A_0 x(t) + \left[ \underline A_1 +  B\Sigma_1 C\right]x(t-\tau)$. 
As stated in Lemma \ref{lem:delayedpositivesys}, positivity of a delayed linear system requires the system matrix $\underline A_0$ to be Metzler and $\underline A_1+B\Sigma_1 C \ge 0$; therefore, the necessity follows.

\emph{Sufficiency.} Suppose that $\underline A_0$ is Metzler, $\underline A_1 + B\Sigma_1 C \ge 0$, and $\Phi \in \text{Sector}[\Sigma_1, \Sigma_2]$. Consider the $j$-th component of the state equation:
\begin{align*}
&\dot{x}_j(t) = A_{0_{jj}} x_j(t) + A_{1_{jj}} x_j(t-\tau) +\\
& \sum_{k = 1, k \neq j}^{n} A_{0_{jk}} x_k(t) + A_{1_{jk}} x_k(t-\tau) + (B\Phi(Cx(t-\tau)))_j.
\end{align*}

Assume $x(0),x_{h,1} \in \mathbb{R}^n_+$ and suppose, for contradiction, that some component $x_j(t)$ becomes negative. Therefore, continuity of solutions suggests that there exists a first time $t_1 \geq 0$ such that $x_j(t_1) = 0$, with $x_k(t_1) \geq 0$ for all $k \neq j$ and $Cx(t-\tau) \geq 0$ for $t \in [-\tau, t_1]$. By the sector condition, partially ordered sets, and nonnegativity of $B$, we obtain:
\vspace{-.3cm}

{\small \begin{align*}
&\dot{x}_j(t_1) \geq \underline A_{0_{jj}} x_j(t_1) + (\underline A_1 + B\Sigma_1 C)_{jj} x_j(t_1-\tau) + \\ & \sum_{\substack{k = 1 , k \neq j}}^{n} \underline A_{0_{jk}} x_k(t_1) + (\underline A_1 + B\Sigma_1 C)_{jk} x_k(t_1-\tau).
\end{align*}}

Since $x_j(t_1) =0$, $\underline A_0$ is Metzler,  $\underline A_1 + B\Sigma_1 C\ge 0$, and $x_k(t) \geq 0$, it follows that $\dot{x}_i(t_1) \geq 0$, contradicting the assumption that $x_j$ becomes negative. Hence, the trajectory remains in $\mathbb{R}^n_+$ for all $t \geq 0$.
\end{proof}
% (i) We denote by $\mathcal{X}$ a neighborhood of the origin where the NN is evaluated. The NN is locally input-bounded on $\mathcal{X}$ so that preactivations lie in a known interval (obtained by interval/IBP propagation or data), enabling local sector/QC descriptions of $\Phi$. (ii) The plant is stabilizable/observable as needed by our certificates.
% \end{assumption}
% (iii) When robustness margins/ROA are computed, we use a quadratic Lyapunov function $V(x)=x^\top Px$ and the ellipsoid $\mathcal{E}(P):=\{x: x^\top Px\le 1\}$.
For analysis purpose we consider three risk configurations:
\begin{enumerate}[label=(C\arabic*)]
    \item \textbf{Interval uncertainty, no delay:} the plant has interval (elementwise) parameter uncertainty and no time delays.
    \item \textbf{Delay only:} constant time delays in the state and input channels, with no parametric uncertainty.
    \item \textbf{Combined case:} simultaneous interval (elementwise) uncertainty and state/input time delays.
\end{enumerate}

% We break the analysis into three cases.
% % \vspace{-.05cm}
% \begin{enumerate}
%     \item system with interval uncertainty and no delay
%     \item system with state and input delay and no uncertainty
%     \item system with both interval uncertainty and delay
% \end{enumerate}

\subsection*{C1 - Interval Uncertainty, No Delay NN Feedback Loops}
Here we consider the system \eqref{eq:wholesys} with $\tau_i = 0$:
\begin{equation}\label{eq:nodelay}\small
    \dot x(t) = [\underline{A}_0,\overline{A}_0]\,x(t)
    + \sum_{i=1}^{\ell} [\underline{A}_i,\overline{A}_i]\,x(t)
    + B_i \Phi\big(Cx(t)\big), 
\end{equation}
\begin{theorem}\label{the:the1}
    Consider the system \eqref{eq:nodelay} with the nonlinear feedback characterized by an FFNN as in \eqref{eq:ffnn} which is $\Gamma-$sector bounded in $[\gamma_1,\gamma_2]$. Under assumption \ref{assump1}, the NN feedback loop is locally exponentially stable in the set $\Gamma$ if $\underline{A}+B\gamma_1C$ is Metzler and $\overline A+B\gamma_2C$ is Hurwitz, where $\overline A=\sum_{i=0}^l \overline A_i, \underline A=\sum_{i=0}^l \underline A_i$, and $\gamma_1,\gamma_2$ are calculated using \eqref{eq:localsectorbound}.
\end{theorem}

\begin{proof}
Assume $x_0 \geq 0$. Given $\dot x \ge \underline{A}+B\gamma_1C$, by Lemma \ref{lem:lurepositivity} (with no delay channel), the Metzler condition on $\underline{A}+B\gamma_1C\geq 0$ ensures $x(t) \geq 0$ for all $t \geq 0$.
Define $M := \overline{A} + B\gamma_2 C$, which is Hurwitz and Metzler due to the Metzler lower sector. By Lemma \ref{lem:lemav}, there exists a positive vector $v > 0$ and scalar $\varepsilon > 0$ such that $v^\top M \leq -\varepsilon v^\top$.

Consider the Lyapunov function $V(x) = v^\top x$.
Since $v > 0$ and $x(t) \geq 0$, we have $V(x) > 0$ for all $x \neq 0$ and $V(0) = 0$.

Computing the time derivative along system trajectories:
$$\dot{V}(x) = v^\top \dot{x} = v^\top \left[Ax(t) + B\Phi(Cx(t))\right]\le v^\top M x.$$
Therefore
$\dot{V}(x) \leq v^\top M x \leq -\varepsilon v^\top x = -\varepsilon V(x).$
This implies $V(x(t)) \leq V(x(0))e^{-\varepsilon t}$, establishing local exponential stability in $\Gamma$.
\end{proof}

\subsection*{C2 - Delayed Only NN feedback loops}

\begin{theorem}\label{thm:kraso_pos_aiz_delay}
Consider the delayed Lur'e system
\begin{equation}\label{eq:lure_delay_TAC}\small
\dot x(t)=A_0 x(t)+\sum_{i=1}^\ell A_i x(t-\tau_i)+B_i\,\Phi\!\big(Cx(t-\tau_i)\big),
\end{equation}
where the feedback $\Phi(Cx)$ is characterized by an FFNN as in \eqref{eq:ffnn}. Suppose there exists a local $\Gamma-$set including the origin in which $\Phi(Cx)$ is sector bounded in $[\gamma_1,\gamma_2]$.
Assume

\textbf{(I)} (\emph{Internal positivity at the lower sector}) $A_0$ is Metzler, $B,C\ge 0$, and
\(
A_i+B\gamma_i C\ \ge\ 0 .
\)

\textbf{(II)} (\emph{Upper-sector DC Hurwitz}) The DC matrix at the upper sector
\(
H :=\sum_{i=0}^nA_i+B\gamma_2 C
\)
is Hurwitz.

Then, for every delay $\tau\ge 0$ 
the closed-loop system \eqref{eq:lure_delay_TAC} is a positive system on $\mathbb{R}^n_+$ and is locally exponentially stable from nonnegative histories.
\end{theorem}

\begin{proof}
For notational simplicity, we present the proof for a single delay channel (\(\ell=1\)); the extension to multiple channels follows analogously.

\emph{Positivity.}
The positivity of the nonlinear closed loop system follows from Lemma \ref{lem:lurepositivity} by fixing $A_0$ and $A_1$.
% Fix a nonnegative history $x|_{[-\tau,0]}\in \mathcal C([-\tau,0];\mathbb{R}^n_+)$.
% Because $A_0$ is Metzler and $B,C\ge 0$, and since $\Phi(\cdot,t)\ge \Sigma_1(\cdot)$ on $\mathbb{R}^p_+$, we have, for all $t\ge 0$ with $x(t-\tau)\ge 0$,
% \[
% A_1 x(t-\tau)+B\,\Phi\!\big(Cx(t-\tau),t\big)\ \ge\ (A_1+B\Sigma_1C)\,x(t-\tau)\ \ge\ 0 .
% \]
% Hence the vector field points into the cone at the boundary and $\mathbb{R}^n_+$ is forward-invariant; the closed loop is internally positive.
% % (positivity is understood as invariance of $\mathbb{R}^n_+$ for nonnegative initial data). 
% :contentReference[oaicite:1]{index=1} :contentReference[oaicite:2]{index=2}
% (Notice that no monotonicity of $\Phi$ is assumed; the sector is imposed only on $\mathbb{R}^p_+$. 
% :contentReference[oaicite:3]{index=3})

\smallskip
\emph{stability certificate.}
Since $H$ is Hurwitz and Metzler (from condition (I)), there exists a positive vector $v > 0$ and scalar $\varepsilon>0$ such that
\begin{equation}\label{eq:cop_cert}\small
H^\top v\ \le\ -\,\varepsilon\,v .
\end{equation}
% This is the standard linear copositive Lyapunov certificate for stable positive linear systems. :contentReference[oaicite:4]{index=4}

% \smallskip
Assume the function $V(x)=v^\top x$.
Let $M:=A_1+B\gamma_2 C\ge 0$ and $w^\top:=v^\top M\ge 0$. Along positive trajectories of \eqref{eq:lure_delay_TAC}, using the upper sector bound inside the derivative estimate,
\vspace{-.3cm}

{\small\begin{align}\label{eq:inequality}
\dot V(t)
&= v^\top A_0 x(t)+v^\top A_1 x(t-\tau)+v^\top B\,\Phi\!\big(Cx(t-\tau)\big) \nonumber\\
&\le v^\top A_0 x(t)+v^\top\big(A_1+B\Sigma_2 C\big)x(t-\tau)\nonumber\\
&= v^\top H x(t)\;-\;w^\top x(t)\;+\;w^\top x(t-\tau)\nonumber\\
&\le -\,\varepsilon\,V(t)\;-\;w^\top x(t)\;+\;w^\top x(t-\tau),
\end{align}}

where the last inequality uses \eqref{eq:cop_cert} and positivity of $x(\cdot)$.

Define
\(
\beta\ :=\ \max_{j}\frac{(M^\top v)_j}{v_j}.
\)
Thus, 
\begin{equation}\label{eq:betadeefect}
w^\top x\le \beta\,V(x), \quad \forall x\in\mathbb{R}^n_+.
\end{equation}
% Hence
% \begin{equation}\label{eq:V_scalar_delay}
% \dot V(t)\ \le\ -\,\varepsilon\,V(t)\ +\ \beta\,V(t-\tau).
% \end{equation}

% \smallskip
% \emph{Krasovskii functional and exponential decay.}
Introduce the copositive Lyapunov–Krasovskii functional
\vspace{-.2cm}

{\small \begin{equation*}
W(t)\ :=\ V\big(x(t)\big)\ +\ \int_{t-\tau}^{t} w^\top x(s)\,ds\ \ \ge 0 .
\end{equation*}}

Differentiating and using \eqref{eq:inequality} yields
\[
\dot W(t)\ =\ \dot V(t)+w^\top x(t)-w^\top x(t-\tau)\ \le\ -\,\varepsilon\,V(t).
\]
As a result, the stability is guaranteed. To prove exponential decay, we use \eqref{eq:betadeefect} to write
\vspace{-.1cm}

{\small\[
W(t)\ \le\ V\big(x(t)\big)\ +\ \beta\tau \sup_{u\in[t-\tau,t]}V\big(x(u)\big).
\]}

At times when $V\big(x(t)\big)=\sup_{u\in[t-\tau,t]}V\big(x(u)\big)$ (which hold on a dense set by continuity), we have
$W(t)\le (1+\beta\tau)\,V\big(x(t)\big)$, hence
\vspace{-.1cm}

{\small\[
\dot W(t)\ \le\ -\,\varepsilon\,V(t)\ \le\ -\,\frac{\varepsilon}{1+\beta\tau}\,W(t).
\]}

Therefore,
\(
W(t)\ \le\ W(0)\,\exp\!\Big(-\frac{\varepsilon}{1+\beta\tau}\,t\Big)\qquad\forall\,t\ge 0,
\)
and consequently $V\big(x(t)\big)\to 0$ exponentially. Since the system is positive, this implies exponential stability from nonnegative histories inside $\Gamma-$set for every $\tau\ge 0$.
\end{proof}

% \begin{remark}
% (i) Monotonicity of $\Phi$ is not required. 
% (ii) Assumption (H) is a “DC-sum” condition familiar from Metzler retarded systems: for positive delay systems, Hurwitzness of the sum is equivalent to (delay-independent) stability of the associated linear member, which is consistent with our choice of $H$.
% % :contentReference[oaicite:6]{index=6}
% (iii) Existence of $v\gg 0$ with $H^\top v<0$ is the usual copositive Lyapunov characterization of Hurwitz Metzler matrices.
% % :contentReference[oaicite:7]{index=7}
% \end{remark}

\subsection*{C3 - Delayed uncertain NN feedback loops}
In the third case we combine the previous two sections and their results.
\begin{theorem}\label{thm:kraso_pos_aiz_delayinterval}
Consider the delayed Lur'e system
\begin{equation}\label{eq:lure_uncertaindelay_TAC}\small
\dot{x}(t)=\left[\underline{A}_0, \bar{A}_0\right] x(t)+\sum_{i=1}^l\left[\underline{A}_i, \bar{A}_i\right] x\left(t-\tau_i\right)+B_i \,\Phi\!\big(Cx(t-\tau_i)\big)
\end{equation}
% with $A_0,A_1\in\mathbb{R}^{n\times n}$, $B\in\mathbb{R}^{n\times m}$, $C\in\mathbb{R}^{p\times n}$. 

% \begin{enumerate}
\textbf{(I)} (\emph{Sector on the positive cone}) Assume there exists a local $\Gamma$-set including the origin in which the NN $\Phi(Cx)$ is sector bounded in $[\gamma_1,\gamma_2]$.

\textbf{(II)} (\emph{Internal positivity at the lower endpoint}) $\underline A_0$ is Metzler, $B,C\ge 0$, and
\(
\underline A_1+B\gamma_1 C\ \ge\ 0 .
\)

\textbf{(III)} (\emph{Upper-endpoint DC Hurwitz}) The DC matrix at the upper endpoint
\(
H\ :=\overline A_0+\overline A_1+B\gamma_2 C
\)
is Hurwitz.

% \end{enumerate}
Then, for every delay $\tau\ge 0$
% and for every (possibly time-varying, nonmonotone) $\Phi\in\mathrm{Sector}[\Sigma_1,\Sigma_2]$, 
the closed-loop system \eqref{eq:lure_uncertaindelay_TAC} is a positive system on $\mathbb{R}^n_+$ and is globally exponentially stable from nonnegative histories.
\end{theorem}

% One possible uncertainty structure is considered to be of interval type leading to 

% $$
% \dot{x}(t)=\left[\underline{A}_0, \bar{A}_0\right] x(t)+\sum_{i=1}^l\left[\underline{A}_i, \bar{A}_i\right] x\left(t-\tau_i\right)+B u(t)
% $$

% and the following result can be stated.
% Theorem 4: The Metzlerian interval delay system (22) is robustly stable if and only if the following single Metzlerian system without delay, i.e. $\dot{x}=\bar{A} x+B u, \quad \bar{A}=\sum_{i=0}^l \bar{A}_i$ is asymptotically stable, and the stability can be verified by one of the equivalent conditions of Lemma 2.
The proof follows from the proofs of Theorems \ref{the:the1} and \ref{thm:kraso_pos_aiz_delay}.

\paragraph*{Note} Due to space limitations, the region of attraction analysis regarding the $\Gamma-$set is not included here. A detailed treatment of this aspect can be found in our earlier work \cite{hedesh2025local}.

\begin{figure}[t]
  \centering
  \includegraphics[width=.9\columnwidth]{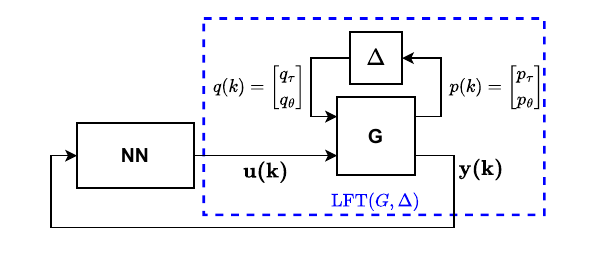}%
  \caption{Linear fractional transformation (LFT) model of the NN-feedback-loop with delay and interval uncetainty.\vspace{-.3cm}}
  \label{fig:LFT}
\end{figure}
\subsection{IQC-Based Method}\label{sec:delay-interval-IQC}
In this subsection we adapt a state-of-the-art IQC framework for NN-in-the-loop verification \cite{yin2021stability} and tailor it to our risk model with delays and interval uncertainty to enable a fair comparison with our positivity-based approach. To align with \cite{yin2021stability} and to model delays in a rational state space filter, we work in discrete time and detail the necessary adaptations step by step. For completeness, the extension to continuous time is shortly outlined in a subsequent remark.

% \subsection{Delay \& Interval Uncertainty via LFT+IQC (Discrete Time)}

We expose the two uncertainty mechanisms—discrete delays and entrywise interval parametric uncertainty—through a single LFT interconnection as shown in \ref{fig:LFT}. Denote by
\begin{equation}\small
p=\begin{bmatrix}p_{\tau}\\ p_{\theta}\end{bmatrix},\qquad
q=\begin{bmatrix}q_{\tau}\\ q_{\theta}\end{bmatrix},
\end{equation}
the stacked uncertainty port signals for delay ($\tau$) and interval ($\theta$), respectively. The plant $G$ is written in the standard form
% \begin{subequations}
{\small\begin{align}\label{eq:plant-LFT}
x_{k+1} &= A_G\,x_k + B_{G1}\,q_k + B_{G2}\,u_k,\nonumber\\
p_k     &= C_G\,x_k + D_{G1}\,q_k + D_{G2}\,u_k,\\
u_k     &= \Phi(x_k) \quad \text{(optional controller).}\nonumber
\end{align}}
% \end{subequations}
% Here $x\in\mathbb{R}^n$, $u\in\mathbb{R}^m$.

The input-output relations of each uncertainty class are represented by an IQC. 
Specifically, the IQC is defined through a filter $\Psi$ acting on the input $p$ and output $q$, together with a constraint on the resulting signal $r$. 
The filter $\Psi$ is described by the LTI system
\vspace{-.4cm}

{\small\begin{align}\label{eq:IQCfiltergeneral}
\psi(k+1) &= A_{\Psi}\psi(k) + B_{\Psi 1}p(k) + B_{\Psi 2}q(k), \nonumber \\
r(k) &= C_{\Psi}\psi(k) + D_{\Psi 1}p(k) + D_{\Psi 2}q(k), \\
\psi(0) &= 0, \nonumber
\end{align}}

where $\psi(k)\in\mathbb{R}^{n_\psi}$ is the filter state, $r(k)\in\mathbb{R}^{n_r}$ is the output, and $A_{\Psi}$ is Schur. All matrices are of compatible dimensions.
A bounded, causal operator $\Delta$ is said to satisfy the time-domain IQC associated with $(\Psi,M)$ if, for all signals $p$ and $q=\Delta(p)$ and for all $N$, the prescribed inequality holds.
\begin{equation}\small
\sum_{k=0}^N r(k)^{\top} M r(k) \geq 0.
\end{equation}

We now specify all blocks in three steps. We define the blocks of delay, interval uncertainty, and their combination in turn.

\subsection*{a) Delayed system LFT and IQC}
Suppose there are $\ell$ delay channels with integer delays $d_i$. For each $i=1,\dots,\ell$,
\vspace{-.2cm}

{\small\[
p_{\tau}^{(i)}(k)=\begin{bmatrix}x_k\\ u_k\end{bmatrix}\in\mathbb{R}^{n+m},\quad
q_{\tau}^{(i)}(k)= p_{\tau}^{(i)}(k-d_i) = \begin{bmatrix}x_{k-d_i}\\ u_{k-d_i}\end{bmatrix}
\]}
Stack
{\small\[
p_{\tau}=\operatorname{col}\!\big(p_{\tau}^{(1)},\dots,p_{\tau}^{(\ell)}\big)\in\mathbb{R}^{\ell(n+m)},\quad
q_{\tau}=\operatorname{col}\!\big(q_{\tau}^{(1)},\dots,q_{\tau}^{(\ell)}\big).
\]}
Write the nominal delayed plant without uncertainty as
\begin{equation}\label{eq:nominal-delayed}\small
x_{k+1} \;=\; A_0^c x_k \;+\; \sum_{i=1}^{\ell} A_i^c\,x_{k-d_i} \;+\; \sum_{i=1}^{\ell} B_i\,u_{k-d_i},
\end{equation}
where the superscript $^c$ in $A_0^c$ and $A_1^c$ denote the nominal (center) system without uncertainty. Place all nominal delayed terms into $B_{G1,\tau}$ and generate $p_{\tau}$ from $(x,u)$ via
\[
A_{G,\tau}=A_0^c,\quad
B_{G1,\tau}=\big[\,A_1^c\ \ B_1\ \ \cdots\ \ A_{\ell}^c\ \ B_{\ell}\big],\quad
B_{G2,\tau}=0,\]
\begin{equation}\label{eq:AG-BG-CG-DG-delay}\small
C_{\tau}=\mathbf{1}_{\ell}\otimes\begin{bmatrix}I_n\\ 0_{m\times n}\end{bmatrix},\quad D_{G1,\tau}=0, \quad
D_{G2,\tau}=\mathbf{1}_{\ell}\otimes\begin{bmatrix}0_{n\times m}\\ I_m\end{bmatrix},
\end{equation}
so that $p_{\tau} = C_{\tau}x + D_{G2,\tau}u = \mathbf{1}_{\ell}\otimes\begin{bmatrix}x_k&u_k\end{bmatrix}^\top$. Here $\mathbf{1}_{\ell}$ is the $\ell$-vector of ones and $\otimes$ denotes the Kronecker product.

\subsection*{IQC for the delay}
For each delay channel of size $n_\tau=n+m$, use the $d_i$-step delay-line filter
% \begin{subequations}
\begin{align}\label{eq:delay-filter}
\psi_{k+1}^{(i)} &= (J_{d_i}\!\otimes I_{n_\tau})\,\psi_k^{(i)} + (e_1\!\otimes I_{n_\tau})\,p_{\tau}^{(i)}(k),\nonumber\\
r_{\tau}^{(i)}(k) &= (e_{d_i}^{\top}\!\otimes (-I_{n_\tau}))\,\psi_k^{(i)} + I_{n_\tau}\,q_{\tau}^{(i)}(k),
\end{align}
% \end{subequations}
with $J_{d_i}$ the down-shift matrix, $e_j$ the $j$-th standard basis vector, and $\psi^{(i)}\in\mathbb{R}^{d_i n_\tau}$. The hard time-domain IQC is
\begin{equation}\label{eq:delay-IQC}\small
\sum_{k=0}^{N} \big(r_{\tau}(k)\big)^{\!\top} M_{\tau}\, r_{\tau}(k)\ \ge 0,\qquad
M_{\tau} = I,\qquad \forall N\ge0,
\end{equation}
where $r_{\tau}=\operatorname{col}\big(r_{\tau}^{(1)},\dots,r_{\tau}^{(\ell)}\big)$ with $r_{\tau}^{(i)} = q_{\tau}^{(i)}(k) - p_{\tau}^{(i)}(k-~d_i) \equiv 0$ satsifying exactly the equality in condition \eqref{eq:delay-IQC}.

\subsection*{b) Interval Uncertainty LFT and IQC}
We expose the interval uncertainty as a separate port $(p_{\theta},q_{\theta})$ by moving the uncertainty part to an additive term $E\,q_{\theta,k}$ as:
\begin{equation}\label{eq:uncertain-delayed}\small
x_{k+1} \;=\; A_0^c x_k + \sum_{i=1}^{\ell} A_i^c\,x_{k-d_i} + \sum_{i=1}^{\ell} B_i\,u_{k-d_i} + E\,q_{\theta,k}.
\end{equation}
Rewrite $\big[\underline A_0,\overline A_0\big]$ and $\big[\underline A_1,\overline A_1\big]$ as affine/interval structure
{\small\begin{equation*}
A_0(\theta) = A_0^c + \sum_{j=1}^{s_0} \theta_{0j}\,\hat A_{0j},\quad
A_i(\theta) = A_i^c + \sum_{j=1}^{s_i} \theta_{ij}\,\hat A_{ij},\;i=1,..,\ell,
\end{equation*}}

with $s_0,s_i$ being the number of independent scalar uncertainty channels to parametrize the uncertainty in 
$A_0^c,A_i^c$, and the independent scalars $\theta_{\cdot j}$ constrained in a box specified by an IQC as shown later.
Collect all basis matrices in
\[
E := \big[\,\hat A_{0,1}\ \cdots\ \hat A_{0,s_0}\ \ \hat A_{1,1}\ \cdots\ \hat A_{\ell,s_{\ell}}\,\big]\in\mathbb{R}^{n\times n s_{\rm tot}},
\]
where $s_{\rm tot}:=s_0+\cdots+s_{\ell}$, and form the signal that repeats the appropriate state samples:
\begin{equation}\label{eq:z-theta}\scriptsize
p_{\theta,k} = \operatorname{col}\Big(\underbrace{x_k,\dots,x_k}_{s_0\ \text{times}},\ \underbrace{x_{k-d_1},\dots,x_{k-d_1}}_{s_1},\ \dots,\ \underbrace{x_{k-d_{\ell}},\dots,x_{k-d_{\ell}}}_{s_{\ell}}\Big).
\end{equation}
Let $\Theta_k=\operatorname{diag}(\theta_{0,1}I_n,\dots,\theta_{\ell,s_{\ell}}I_n)$ and define
\[
q_{\theta,k} = \Theta_k p_{\theta,k},\qquad\text{so that}\qquad x_{k+1} \gets x_{k+1} + E\,q_{\theta,k}.
\]
% We take the interval port to be
% \[
% p_{\theta}=z,\qquad q_{\theta}=w_{\theta}.
% \]

Thus, following \eqref{eq:plant-LFT}, the updated blocks from \eqref{eq:AG-BG-CG-DG-delay} that generate $p=\operatorname{col}(p_{\tau},p_{\theta})$ and implement $q=\operatorname{col}(q_{\tau},q_{\theta})$ are
\[A_G=A_0^c,\quad
B_{G1}=\big[\,A_1^c\ \ B_1\ \ \cdots\ \ A_{\ell}^c\ \ B_{\ell}  \ \ E\big],\quad
B_{G2}=0,\]
\begin{equation}\label{eq:CG-DG-full}\small
C_G=\begin{bmatrix} C_{\tau}\\ C_{\theta}\end{bmatrix},\qquad
D_{G1}=\begin{bmatrix} 0\\ D_{\theta\leftarrow\tau}\end{bmatrix},\quad
D_{G2}=\begin{bmatrix} D_{G2,\tau}\\ 0\end{bmatrix},
\end{equation}
where $C_{\theta} =\begin{bmatrix}\mathbf{1}_{s_0}\otimes I_n\\ \mathbf{0}\end{bmatrix}$ stacks $x_k$ exactly $s_0$ times, and $D_{\theta\leftarrow\tau} = \mathbf{1}_{s_L}\otimes \begin{bmatrix}
     I_n & \mathbf{0}
\end{bmatrix}$ selects the delayed $x_{k-d_i}$ components from each $q_{\tau}^{(i)}$ and repeats them $s_L = s_1+\dots+s_\ell$ times, in the same order as in \eqref{eq:z-theta}.

\subsection*{IQC for the interval block}
We use a memoryless filter that outputs the pair $(p_{\theta},q_{\theta})$ directly, and a Symmetric Box multiplier.
\vspace{-.3cm}

{\small \begin{align}\label{eq:IQCfilterinterval}
&\psi(k+1) = \mathbf{0}, \nonumber \\
&r_\theta(k) = \mathbf{0}\psi(k) + \begin{bmatrix}
    \mathbf I \\ \mathbf 0
\end{bmatrix}p_\theta(k) + \begin{bmatrix}
    \mathbf 0\\\mathbf I
\end{bmatrix}q_\theta(k).
\end{align}}

Symmetric Box: $q_\theta=\theta \,p_\theta$ with $|\theta_j|\le \rho_j$ gives

{\small\[
M_{{\rm \theta},j}(\rho_j)=\begin{bmatrix}\rho_j^2 I & 0\\ 0 & -I\end{bmatrix}\succeq 0.
\]}

Stack with nonnegative multipliers $\mu_j\ge 0$:
\begin{equation}\label{eq:M-theta}
M_{\theta}(\mu)=\operatorname{blkdiag}\big(\mu_1 M_{\theta,1},\dots,\mu_{s_{\rm tot}} M_{\theta,s_{\rm tot}}\big)\succeq 0.
\end{equation}
% where $M_j$ is the Symmetric Box block for channel $j$.

\subsection*{c) Composite delay and interval LFT \& IQC.}
For constructing the combined IQC, let:
\begin{equation}\label{eq:combinedIQCs}\small
\Psi_{\Delta}=\operatorname{blkdiag}(\Psi_{\tau},\Psi_{\theta}) \qquad M_{\Delta}=\operatorname{blkdiag}(M_{\tau},M_{\theta}(\mu)).
\end{equation}
% With $\psi$ the stack of all delay-line states (the interval block is static),
Define the extended state
\(
\zeta:=\begin{bmatrix}x\\ \psi\end{bmatrix}.
\)
The assembly gives the extended realization
\begin{subequations}\label{eq:extended-ABCD}
\begin{align}\small
\zeta_{k+1} &= A\,\zeta_k + B\,\begin{bmatrix}q_k\\ u_k\end{bmatrix},\\
r_k         &= C\,\zeta_k + D\,\begin{bmatrix}q_k\\ u_k\end{bmatrix},
\end{align}
\end{subequations}
% with
{\small\begin{equation*}\label{eq:ABCD-blocks}
\begin{aligned}
A&=\begin{bmatrix} A_G & 0\\ B_{\Psi 1} C_G & A_{\Psi}\end{bmatrix},\quad
B=\begin{bmatrix} B_{G1} & B_{G2}\\ B_{\Psi 1} D_{G1}+B_{\Psi 2} & B_{\Psi 1} D_{G2}\end{bmatrix}\\
C&=\begin{bmatrix} D_{\Psi 1} C_G & C_{\Psi}\end{bmatrix},\quad
D=\begin{bmatrix} D_{\Psi 1} D_{G1}+D_{\Psi 2} & D_{\Psi 1} D_{G2}\end{bmatrix}
\end{aligned}
\end{equation*}}

where $(A_{\Psi},B_{\Psi},C_{\Psi},D_{\Psi})$ are those of $\Psi_{\Delta}$,
and the columns of $(B,C,D)$ are ordered consistently with $q=\operatorname{col}(q_{\tau},q_{\theta})$.

% Requires: \usepackage{amsmath,amssymb,amsthm}
% (Optional) theorem setup:
% \newtheorem{theorem}{Theorem}

\begin{theorem}\label{thm:LyapIQC}
Consider the extended closed-loop system as in \eqref{eq:extended-ABCD}, and the composite IQC \eqref{eq:combinedIQCs}.
If there exist \(P\succ 0\), multipliers \(\mu_j\ge 0\), and some \(\varepsilon>0\) such that
\begin{equation}\label{eq:main-LMI}\small
\begin{bmatrix}
A^\top P A - P & A^\top P B\\
B^\top P A     & B^\top P B
\end{bmatrix}
\;+\;
\begin{bmatrix}C \\ D\end{bmatrix}
M_\Delta
\begin{bmatrix}C & D\end{bmatrix}
\ \preceq\ -\,\varepsilon I,
\end{equation}
then the origin is locally asymptotically stable for the given delays \(\{d_i\}\) and for all uncertainty realizations consistent with the interval box bounds. 
Moreover, every sufficiently small sublevel set \(\{\zeta:V(\zeta)\le\rho\}\), contained in the region 
where the IQCs are valid, is forward invariant.
\end{theorem}

\begin{proof}
Let \(V(\zeta):=\zeta^\top P\zeta\). Left/right multiplying \eqref{eq:main-LMI} by
\(\big[\zeta^\top\ \ q^\top\ \ u^\top\big]\) gives, at each time \(k\),
{\small\[
V(\zeta_{k+1})-V(\zeta_k)
+ r_k^\top\!\big(M_\Delta\big) r_k
% + r_k^\top M_\theta(\mu)\, r_k
\le-\,\varepsilon\,\Big\|\!\begin{bmatrix}\zeta_k\\ q_k\\ u_k\end{bmatrix}\!\Big\|^2\!.
\]}

Summing from \(k=0\) to \(N\) yields
\vspace{-.5cm}

\begin{equation}\small
V(\zeta_{N+1})-V(\zeta_0)
+ \sum_{k=0}^{N}\!\Big(
r_k^\top(M_\Delta) r_k
% + r_k^\top M_\theta(\mu) r_k
\Big)
\le-\varepsilon \sum_{k=0}^{N}\!\Big\|\!\begin{bmatrix}\zeta_k\\ q_k\\ u_k\end{bmatrix}\!\Big\|^2.
\end{equation}

Since \(M_\tau,M_\theta(\mu)\succeq 0\) and \(V(.)\ge 0\), letting \(N\to\infty\) gives
\(\sum_{k=0}^{\infty}\!\big\|\![\zeta_k^\top\ q_k^\top\ u_k^\top]^\top\big\|^2 \le V(\zeta_0)/\varepsilon<\infty\),
hence \(\|[\zeta_k^\top\ q_k^\top\ u_k^\top]^\top\|\to 0\).
% and, in particular, \(\zeta_k\to 0\).
If \(V(\zeta)\le\rho\) ensures the satisfaction of the IQCs, then
\(V(\zeta_{k+1})\le V(\zeta_k)\) on the set, and for \(\rho\) small enough, the sublevel set is forward invariant.
\end{proof}

\begin{remark}
If an NN controller $u=\pi(x)$ is included, its local-sector IQC is added as an extra quadratic term on the right-hand side of \eqref{eq:main-LMI}. A local sector IQC for FFNNs is presented in \cite{yin2021stability}. We shortly explain it below and utilize it to analyze the stability of delayed-uncertain NN feedback loops.
\end{remark}

\paragraph*{Neural network IQCs}
Let $\nu_\phi$ and $\omega_\phi = \phi(\nu_\phi)$ be the stacked inputs/outputs of all
activations at the NN equilibrium $(\nu^\star,\omega^\star)$. If the scalar activation
$\varphi$ satisfies local sector bounds
$[\alpha^{(i)}_{\phi},\beta^{(i)}_{\phi}]$ on $\nu^{(i)}_{\phi}\in [\underline \nu^{(i)},\overline \nu^{(i)}]$
around $\nu^{(i)}_*$, as shown in \ref{fig:tanh-relaxations} and stacked as $(\alpha_\phi,\beta_\phi)$,
then for any $\lambda\in\mathbb{R}^{n_\phi}_+,\; (n_\phi=n_1+n_2+\dots+n_q$ summation of NN layer sizes),
\begin{equation}\small
\begin{bmatrix}
\nu_\phi - \nu^\star \\[2pt] \omega_\phi - \omega^\star
\end{bmatrix}^{\!\top}
\Psi_\phi^\top M_\phi(\lambda)\,\Psi_\phi
\begin{bmatrix}
\nu_\phi - \nu^\star \\[2pt] \omega_\phi - \omega^\star
\end{bmatrix}\;\ge\;0,
\end{equation}
{\small\begin{equation*}
\Psi_\phi =\!
\begin{bmatrix}
\mathrm{diag}(\beta_\phi) & -I_{n_\phi}\\[2pt]
-\mathrm{diag}(\alpha_\phi) & I_{n_\phi}
\end{bmatrix},
\quad
M_\phi(\lambda) =\!
\begin{bmatrix}
0_{n_\phi} & \mathrm{diag}(\lambda)\\[2pt]
\mathrm{diag}(\lambda) & 0_{n_\phi}
\end{bmatrix}.
\end{equation*}}

These NN QCs are added to the composite IQC used for
delays/uncertainty by simply adding an extra term
$R_\phi^\top \Psi_\phi^\top M_\phi(\lambda)\Psi_\phi R_\phi$
to the LMI \eqref{eq:main-LMI}. Here, $R_\phi$ is the selection matrix defined by 
$R_\phi \begin{bmatrix}\zeta \\ q \\ u\end{bmatrix} 
= \begin{bmatrix} v_\phi - v^\star \\[2pt] w_\phi - w^\star \end{bmatrix}$, 
with detailed construction given in~\cite{yin2021stability}.

\paragraph*{Notes on dimensions}
(i) The delay filter has $\sum_{i=1}^{\ell} d_i (n+m)$ states. (ii) The dimension of the interval filter is $\sum_{i=1}^{s_{tot}} n^2$. This can be reduced by introducing a low-dimensional basis that specifies how the uncertainty perturbs $A_0,A_i$. However, such modeling couples the uncertainty channels and may exclude admissible elementwise variations. Thus, the elementwise channels adopted in this paper is the safest choice.
\paragraph*{Extension to continuous time}
Use a rational delay approximation, e.g., Padé and its IQC to obtain an augmented rational LTI model, and replace the discrete-time Lyapunov condition in \eqref{eq:main-LMI} with \(A^\top P + P A\); the remaining steps are the same with minimal adaptations.

\section{Example}
\paragraph*{\textbf{C1) Interval uncertainty only}}
To corroborate the positivity-based certificate, we choose an open-loop unstable interval system \eqref{eq:nodelay} with $\underline A_0 = \begin{bmatrix}
    -8 & 2 & 1\\3 & -10 & 2\\1 & 2 & -8
\end{bmatrix},\; \overline A_{0_{ij}} = \underline A_{0_{ij}}+0.5$, $\underline A_1 = 3\,\mathbf{1}\mathbf{1}^\top,\;\overline A_1 = 3.5\,\mathbf{1}\mathbf{1}^\top,\; B=\mathbf{1},\;C=\mathbf{1}^\top$. We close the loop with our FFNN (zero biases) controller $u=\phi(y)$ trained on trajectories of a PID $(P=-0.9,I=-0.1,D=-0.2)$ controlled system. Using \eqref{eq:localsectorbound}, the NN is calculated to lie in the sector $[-3,-2.44]$ over the $\Gamma-$set: $y\in[0,4.5]$. The sector bounds are shown in Fig. \ref{fig:nn-sector}.

\begin{figure}[t]
  \centering
  \includegraphics[width=0.6\columnwidth]{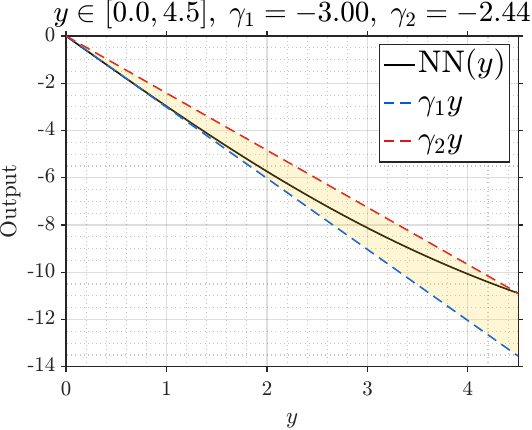}%
  \caption{Using \eqref{eq:localsectorbound}, the local sector bounds over the $\Gamma-$set: $y\in[0,4.5]$ are calculated as $\gamma_1=-3$ and $\gamma_2=-2.44$. As demonstrated $\gamma_1 y\le NN(y) \le \gamma_2 y$.}
  \label{fig:nn-sector}
\end{figure}

% In the delay-free case, we draw four plants uniformly at random from the interval box $(A_0,A_1)\in[\underline A_0,\overline A_0]\times[\underline A_1,\overline A_1]$ and, for each plant, run a Monte-Carlo with $200$ initial conditions sampled uniformly from $[-1.5,1.5]^3$
System simulation results are demonstrated in Fig.~\ref{fig:mc_interval_nn}. Consistent with Theorem \ref{the:the1}, $\underline{A}+B\gamma_1C$ is Metzler and $\overline A+B\gamma_2C$ is Hurwitz and all trajectories decay to the origin.
\begin{figure}[t]
  \centering
  \includegraphics[width=.8\columnwidth]{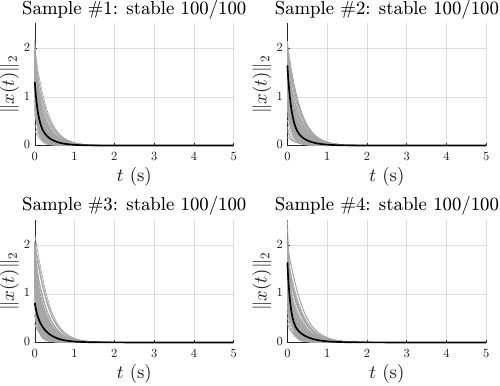}
  \caption{\footnotesize Delay-free interval case with a trained NN controller. Four random plants are sampled from $[\underline A_0,\overline A_0]$ and $[\underline A_1,\overline A_1]$. For each plant, $100$ random initial conditions consistant with $\Gamma-$set $(Cx_0\in[0,4.5])$ are simulated; thin lines show $\|x(t)\|_2$ and the bold line marks a representative (median) trajectory. Subplot titles report the proportion converging to zero.\vspace{-.4cm}}
  \label{fig:mc_interval_nn}
\end{figure}
\paragraph*{\textbf{C2) Delay only}}
We now consider the delayed Lur'e system \eqref{eq:lure_delay_TAC} with
$A_0=\underline{A}_0$ as above,
$A_1=3\,\mathbf{1}\mathbf{1}^\top$,
$B=\mathbf{1}$,
$C=\mathbf{1}^\top$,
and the same trained NN controller. With $A_0$ being Metzler, $B,\ C,\ A_1+B\gamma_1C\ge0$, and $A_0+A_1+B\gamma_2C$ being Hurwitz, Theorem \ref{thm:kraso_pos_aiz_delay} predicts exponential stability. Fig.~\ref{fig:mc_delay_nn} shows the simulation results are consistent with the expectation.
% To empirically corroborate the pre-IQC sector/DC-sum certificate for delay-independent stability,
% we simulate four representative delays $\tau\in\{0.2,\,0.8,\,1.6,\,3.0\}\,$s.
% For each $\tau$, we run a Monte-Carlo with $200$ random constant histories
% $x(t)\equiv x_0$ on $[-\tau,0]$, where $x_0\sim\mathrm{Unif}([-1.5,1.5]^3)$,
% and integrate over $t\in[0,20]$\,s.
% Figure~\ref{fig:mc_delay_nn} shows $\|x(t)\|_2$ (thin gray) and highlights a representative trajectory (bold).
% Across all four delays, trajectories converge to the origin, consistent with the delay-independent stability guarantee for the trained NN in loop.
\begin{figure}[t]
  \centering
  \includegraphics[width=.8\columnwidth]{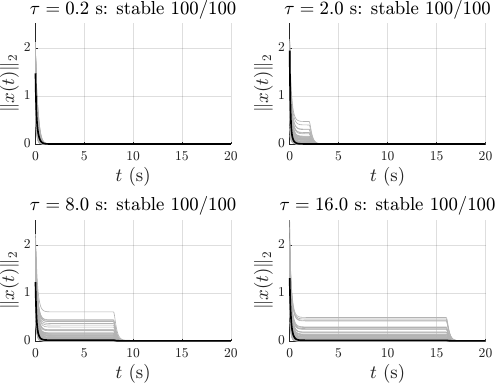}
  \caption{\footnotesize Delay-only case with the trained NN controller.
  Four tiles correspond to $\tau\in\{0.2,2,8,16\}$\,s.
  For each delay, $100$ random constant histories $x(t)\equiv x_0$ on $[-\tau,0]$ with $x_0\!\sim\!\mathrm{Unif}([-1.5,1.5]^3)$ are simulated.
  Thin curves plot $\|x(t)\|_2$; the bold curve marks a representative (median) trajectory.
  Subplot titles report the proportion converging to zero.\vspace{-.1cm}}
  \label{fig:mc_delay_nn}
\end{figure}

\paragraph*{\textbf{C3) Combined delay \& interval uncertainty}}
To illustrate the combined-delay/uncertainty case, we draw $5$ random plants by sampling $(A_0,A_1)$ uniformly from the interval box of the example C1 and $5$ random delays from $\tau\in[0.2,\,3.0]$\,s. For all sampled systems, the conditions of Theorem \ref{thm:kraso_pos_aiz_delayinterval} is satisfied as outlined in examples C1 and C2. As shown in Fig. \ref{fig:outputs_delay_interval}, For each sampled system we simulate from $10$ random constant histories $x_h(t)\!\equiv\!x_0$ inside $\Gamma-$set with the trained NN controller. We distinguish systems with different colors. Across the sampled cases the output trajectories contract toward zero, empirically supporting Theorem \ref{thm:kraso_pos_aiz_delayinterval}.

\paragraph*{\textbf{IQC comparison}}The IQC-based method is also run for two different discretization/delay configurations of the same interval system outlined in example C1, within the same $\Gamma-$set. As shown in Table \ref{tab:disc-delay-runtime}, the IQC-based method is delay/configuration dependent (cannot certify in cases) and performs up to $10^5$\textbf{X} slower. Further comparison notes are explained in Table \ref{tab:method-chars}.

\begin{figure}[t]
  \centering
  \includegraphics[width=.9\columnwidth]{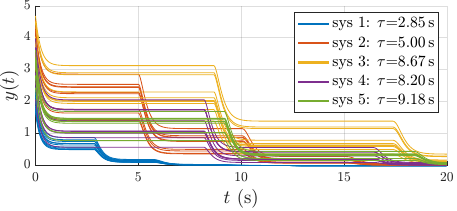}
  \caption{Outputs $y(t)=Cx(t)$ for five sampled $(A_0,A_1,\tau)$ systems; simulated ten positive initial conditions $(y_0\in[0,4.5])$ per system.\vspace{-.2cm}}
  \label{fig:outputs_delay_interval}
\end{figure}
\vspace{-.2cm}
\section{conclusion}
We present a positivity-based, scalable safety-verification framework for autonomous systems with NN controllers subject to two risk sources: time delay and interval-matrix uncertainty. A novel local sector bound for FFNNs underpins linear, delay-independent certificates. We benchmark against a state-of-the-art IQC-based method and, in simulations, achieve faster runtimes while certifying cases the benchmark cannot. Future work will extend the approach to additional risk models and to other NN architectures (e.g., RNNs).

\begin{table}[t]
  \centering
  \scriptsize
  \caption{effect of discretization \& delay on verification runtime.}
  \label{tab:disc-delay-runtime}
  \begin{tabular}{@{}llll@{}}
    \toprule
    Discretization step(s) & Delay(s) & IQC runtime(s) & Positivity runtime(s) \\
    \midrule
    $0.1$ & $0.7$ & \emph{cannot certify} & $5.3\times10^{-5}$ \\
    $0.01$ & $0.07$ & $2.22$  & $5.6\times10^{-5}$ \\
    \bottomrule
  \end{tabular}

  % \vspace{0.25em}
  % \footnotesize Same nonlinear system; rows differ only by discretization \(h\) and delay \(\tau\).
\end{table}

\begin{table}[h]
  \centering
  \scriptsize
  \caption{Method characteristics.}
  \label{tab:method-chars}
  \begin{tabular}{@{}lcc@{}}
    \toprule
    & Positivity-based & IQC-based \\
    \midrule
    Certificate & Metzler/Hurwitz & Lyapunov + IQCs\\
    Solve class      & Linear constraints & SDP \\
    variables & Few (structure-exploiting) & grows with filters/NN/delays \\
    Dimension & $\,n\,$ & $ n(1{+}d_i)+n_\phi+\sum_{i=1}^{s_{tot}} n^2$ \\
    % Memory & $O(n^2)$ & $O(n_\zeta^2)$ \\
    % Per-instance cost & $O(n^3)$ & $O(n_{\text{ext}}^3)+O\!\big(n_{\text{ext}}^2 (s{+}p)\big)$ \\
    \bottomrule
  \end{tabular}
\end{table}

% \vspace{-.1cm}
\begin{spacing}{.78}
\bibliography{Reference}
\end{spacing}
\end{document}